\documentclass[]{article}

\usepackage{longtable}
\usepackage{amssymb}
\usepackage{amsthm}
\usepackage{amsmath}

\newtheorem{theorem}{Theorem}

\newtheorem{proposition}{Proposition}
\usepackage{authblk}
\usepackage[margin=1.25in]{geometry}
\usepackage{appendix}
\usepackage{graphicx}

\begin{document}

\title{Bayesian Inference of Arrival Rate and Substitution Behavior from Sales Transaction Data with Stockouts}
\date{}
\author{Benjamin Letham}
\affil{Operations Research Center, Massachusetts Institute of Technology, bletham@mit.edu}
\author{Lydia M. Letham}
\affil{Department of Electrical Engineering and Computer Science, Massachusetts Institute of Technology, lmletham@mit.edu}
\author{Cynthia Rudin}
\affil{Computer Science and Artificial Intelligence Laboratory and Sloan School of Management, Massachusetts Institute of Technology, rudin@mit.edu}

\maketitle

\begin{abstract}
When an item goes out of stock, sales transaction data no longer reflect the original customer demand, since some customers leave with no purchase while others substitute alternative products for the one that was out of stock. Here we develop a Bayesian hierarchical model for inferring the underlying customer arrival rate and choice model from sales transaction data and the corresponding stock levels. The model uses a nonhomogeneous Poisson process to allow the arrival rate to vary throughout the day, and allows for a variety of choice models. Model parameters are inferred using a stochastic gradient MCMC algorithm that can scale to large transaction databases. We fit the model to data from a local bakery and show that it is able to make accurate out-of-sample predictions, and to provide actionable insight into lost cookie sales.
\end{abstract}

\section{Introduction}
An important common challenge facing retailers is to understand customer preferences in the presence of stockouts. When an item is out of stock, some customers will leave, while others will substitute a different product. From the transaction data collected by retailers, it is challenging to determine exactly what the customer's original intent was, or, because of customers that leave without making a purchase, even how many customers there actually were.

The task that we consider here is to infer both the customer arrival rate, including the unobserved customers that left without a purchase, and the substitution model, which describes how customers substitute when their preferred item is out of stock. Furthermore, we wish to infer these from sales transaction and stock level data, which data are readily available for many retailers. These quantities are a necessary input for inventory management and assortment planning problems.

Stockouts are a common occurrence in some retail settings, such as bakeries and flash-sale retailers \cite{Johnson14}. Not properly accounting for the data truncation caused by stockouts can lead to poor stocking decisions. Na\"{\i}vely estimating demand as the number of items sold underestimates the demand of items that stock out, while overestimating the demand of their substitutes. This could lead the retailer to set the stock for the substitute items too high, while leaving the stock of the stocked-out item too low, potentially losing customers and revenue. 

There are several key features of our model and inference that make it successful in problem settings where prior work in the area has not been. First, prior work has assumed the arrival rate to be constant within each time period \cite{Vulcano12}. Our model allows for arbitrary nonhomogeneous arrival rate functions, which is important for our bakery case study where sales have strong peaks at lunch time and between classes. Second, prior work has required a particular choice model \cite{Anupindi98, Vulcano12, Vulcano14}, whereas our model can incorporate whichever choice model is most appropriate. There are a wide variety of choice models, econometric models describing how a customer chooses one of several alternatives, with different properties and which are applicable in different settings. Third, we model multiple customer segments, each with their own substitution models which can be used to borrow strength across data from multiple stores. Fourth, unlike prior work which has used point estimates, our inference is fully Bayesian. Because we do full posterior inference, we are able to compute the posterior predictive distributions for decision quantities of interest, such as lost sales due to stock unavailability. This allows us to incorporate the uncertainty in estimation directly into uncertainty in our decision quantities, thus leading to more robust decisions.

Our contributions are four-fold. First, we develop a Bayesian hierarchical model that uses the censoring caused by stockouts and their induced substitutions to gain useful insight from transaction data. Our model is flexible and powerful enough to be useful in a wide range of retail settings. Second, we show how recent advances in MCMC for topic models can be adapted to our model to provide a sampling procedure that scales to large transaction databases. Third, we provide a simulation study which shows that we can recover the true generating values and which demonstrates the scalability of the inference procedure. Finally, we make available actual retail transaction data from a bakery\footnote{Data are available at http://github.com/bletham/bakery} and use these data for a case study showing how the model and sampling work in a real setting. In the case study we evaluate the predictive power of the model, and show that our model can make accurate out-of-sample predictions whereas the baseline method cannot. We finally show how the methods developed here can be useful for decision making by producing a posterior predictive distribution of the bakery's lost sales due to stock unavailability.

\section{The Generative Model}
We begin by introducing the notation that we use to describe the observed data. We then introduce the nonhomogeneous model for customer arrivals, followed by a discussion of various possible choice models. Section \ref{sec:mixtures} discusses how multiple customer segments are modeled. Finally, Section \ref{sec:likelihood} introduces the likelihood model and Section \ref{sec:prior} discusses the prior distributions.

\subsection{The Data}
We suppose that we have data from a collection of stores $\sigma=1,\ldots,S$. For each store, data come from a number of time periods $l=1,\ldots,L^{\sigma}$, throughout each of which time varies from $0$ to $T$. For example, in our experiments a time period was one day. We consider a collection of items $i=1,\ldots,n$. We suppose that we have two types of data: purchase times and stock levels. We denote the number of purchases of item $i$ in time period $l$ at store $\sigma$ as $m^{\sigma,l}_i$. Then, we let $\boldsymbol{t}^{\sigma,l}_{i} = \left\{t^{\sigma,l}_{i,1}, \ldots, t^{\sigma,l}_{i,m^{\sigma,l}_i}\right\}$ be the observed purchase times of item $i$ in time period $l$ at store $\sigma$. For notational convenience, we let $\boldsymbol{t}^{\sigma,l} = \left\{\boldsymbol{t}^{\sigma,l}_{i} \right\}_{i=1}^{n}$ be the collection of all purchase times for that store and time period, and let $\boldsymbol{t} = \left\{\boldsymbol{t}^{\sigma,l}\right\}_{\overset{{l=1,\ldots,L^{\sigma}}}{{\sigma=1,\ldots,S}}}$ be the complete set of arrival time data.

We denote the known initial stock level as $N_i^{\sigma,l}$ and assume that stocks are not replenished throughout the time period. That is, $m^{\sigma,l}_i \leq N_i^{\sigma,l}$ and equality implies a stockout. As before, we let $\boldsymbol{N}^{\sigma,l}$ and $\boldsymbol{N}$ represent respectively the collection of initial stock data for store $\sigma$ and time period $l$, and for all stores and all time periods.

Given $\boldsymbol{t}^{\sigma,l}_i$ and $N_i^{\sigma,l}$, we can compute a stock indicator as a function of time. We define this indicator function as
\begin{equation*}
s_i(t \mid \boldsymbol{t}^{\sigma,l}, \boldsymbol{N}^{\sigma,l}) =
\begin{cases}
0 & \textrm{if item }i\textrm{ is out of stock at time }t \\
1 & \textrm{if item }i\textrm{ is in stock at time }t. \\
\end{cases}
\end{equation*}

The generative model for these data will be that customers arrive at the store according to some arrival process. Each customer belongs to a particular segment, and chooses an item to purchase (or no-purchase) based on the preferences of his or her segment and the available stock. When the customer purchases item $i$, the arrival time is recorded in $\boldsymbol{t}_i^{\sigma, l}$. When a customer leaves without making a purchase, for instance because his or her preferred item is out of stock, the arrival time is not recorded. We now present the two main components of this model: the customer arrival process and the choice model.

\subsection{Modeling Customer Arrivals}
We model the times of customer arrivals using a nonhomogeneous Poisson process (NHPP). An NHPP is a generalization of the Poisson process that allows for the intensity to be described by a function $\lambda(t)\geq 0$ as opposed to being constant. We assume that the intensity function has been parameterized, with parameters $\boldsymbol{\eta}^{\sigma}$ potentially different for each store $\sigma$. The most basic parameterization is $\lambda(t \mid \boldsymbol{\eta}^{\sigma}) = \eta^{\sigma}_1$, producing a homogeneous Poisson process of rate $\eta^{\sigma}_1$. As another example, we can produce an intensity function that rises to a peak and then decays by letting
\begin{equation}\label{eq:hill}
\lambda(t \mid \boldsymbol{\eta}^{\sigma}) = \eta^{\sigma}_1 \left(\frac{\eta^{\sigma}_2}{\eta^{\sigma}_3} \right) \left(\frac{t}{\eta^{\sigma}_3} \right)^{\eta^{\sigma}_2-1} \left( 1+\left(\frac{t}{\eta^{\sigma}_3}\right)^{\eta^{\sigma}_2}\right)^{-2},
\end{equation}
which is the derivative of the Hill equation \cite{Goutelle08}.

The posterior of $\boldsymbol{\eta}^{\sigma}$ will be inferred. To do this we use the log-likelihood function for NHPP arrivals, which for arrival times $t_1, \ldots, t_m$ over interval $[0,T]$ is:
\begin{equation*}
\log p(t_1,\ldots,t_m) = \sum_{j=1}^{m} \log \left(\lambda(t_j) \right) -\Lambda(0,T),
\end{equation*}
where $\Lambda(0,T) = \int_{0}^{T} \lambda(t) dt$. Our model can incorporate any integrable rate function. We let $\boldsymbol{\eta} = \left\{\boldsymbol{\eta}^{\sigma}\right\}_{\sigma=1}^S$ represent the complete collection of rate function parameters to be inferred.

\subsection{Models for Substitution Behavior}
Whether or not a customer purchases an item and which item they purchase depends on the stock availability as well as some choice model parameters which we describe below. We define $f_i(s(t),\boldsymbol{\phi}^k, \tau^k)$ to be the probability that a customer purchases product $i$ given the current stock $s(t)$ and choice model parameters $\boldsymbol{\phi}^k$ and $\tau^k$. Then, we denote the no-purchase probability as 
\begin{equation*}
f_0(s(t),\boldsymbol{\phi}^k,\tau^k) = 1 - \sum_{i=1}^n f_i(s(t),\boldsymbol{\phi}^k,\tau^k).
\end{equation*}
The index $k$ indicates the parameters for a particular customer segment, which we will discuss in Section \ref{sec:mixtures}. Posterior distributions for the parameters $\boldsymbol{\phi}^k$ and $\tau^k$ are inferred.

Choice models are econometric models describing a customer's choice between several alternatives, often derived from a utility maximization problem. Different assumptions and utility models lead to different choice models, which ultimately lead to a different form of the purchase probability $f_i(s(t),\boldsymbol{\phi}^k, \tau^k)$. Our model accommodates any choice model for which the purchase probabilities can be expressed as a function of the current stock. We now discuss how several common choice models fit into this framework, and we use these choice models in our simulation and data experiments.

\subsubsection{Multinomial Logit Choice}
The multinomial logit (MNL) is a popular choice model with parameters $\phi^k_1,\ldots,\phi^k_n$ specifying a preference distribution over products, that is, $\phi^k_i \geq 0$ and $\sum_{i=1}^n \phi^k_i = 1$. Each customer selects a product according to that distribution. When an item goes out of stock, substitution takes place by transferring purchase probability to the other items proportionally to their original probability, including to the no-purchase option. This model requires a proportion $\tau^k/(1+\tau^k)$ of arrivals be no-purchases when all items are in stock. The MNL choice probabilities are:
\begin{equation}\label{eq:mnl}
f^{\textrm{mnl}}_i(s(t),\boldsymbol{\phi}^k) = \frac{s_i(t) \phi^k_i}{\tau^k + \sum_{v=1}^n s_v(t) \phi^k_{v}}.
\end{equation}
The MNL model parameter $\tau^k$ is not identifiable when the arrival function is also unknown, a serious disadvantage of this model \cite{Vulcano12}.

\subsubsection{Single-Substitution Exogenous Model}\label{sec:exo}
The exogenous choice model overcomes many of the shortcomings of the MNL model, including the issue of parameter identifiability. According to the exogenous proportional substitution model \cite{Kok07}, a customer samples a first choice from the preference distribution $\boldsymbol{\phi}^k$. If that item is available, he or she purchases the item. If the first choice is not available, with probability $1-\tau^k$ the customer leaves as no-purchase. With the remaining $\tau^k$ probability, the customer picks a second choice according to a preference vector that has been re-weighted to exclude the first choice. Specifically, if the first choice was $j$ then the probability of choosing $i$ as the second choice is $\phi^k_i/\sum_{v \neq j} \phi^k_v$. If the second choice is in stock it is purchased, otherwise the customer leaves as no-purchase. The purchase probability is
\begin{equation}\label{eq:exo}
f^{\textrm{ex}}_i(s(t),\boldsymbol{\phi}^k,\tau^k) = s_i(t)\phi^k_i + s_i(t) \tau^k  \sum_{j=1}^n (1-s_j(t)) \phi^k_j \frac{\phi^k_i}{\sum_{v\neq j} \phi^k_v}.
\end{equation}
Posterior distributions for both $\boldsymbol{\phi}^k$ and $\tau^k$ are inferred.

Allowing for the no-purchase option only in the event of stockouts means that the inferred arrival rate will be that of customers who actually would have purchased an item had all items been in stock. It would be possible for the exogenous model to include a proportion of customers that make no purchase even with full stock, as is required by the MNL model. However, inasmuch as these customers make no contribution to sales regardless of stock, it serves no purpose in the ultimate goal of understanding the effect of stock on sales.

\subsubsection{Nonparametric Choice Model}
Nonparametric models describe preferences as an ordered set of items. Let $\boldsymbol{\phi}^k$ be an ordered subset of the items $\{1,\ldots,n\}$. Customers purchase $\phi^k_1$ if it is in stock. If not, they purchase $\phi^k_2$ if it is in stock. If not, they continue substituting down $\boldsymbol{\phi}^k$ until they reach the first item that is available. If none of the items in $\boldsymbol{\phi}^k$ are available, they leave as a no-purchase. The purchase probability for this model is then $1$ for the first in-stock item in $\boldsymbol{\phi}^k$, and $0$ otherwise.

Because this model requires all customers to behave exactly the same, it is most useful when customers are modeled as coming from a number of different segments $k$, each with its own preference ranking $\boldsymbol{\phi}^k$. This is precisely what we do in our model, as we describe in the next section. For the nonparametric model the rank orders for each segment $\boldsymbol{\phi}^k$ are fixed and it is the distribution of customers across segments that is inferred. We do not generally need to consider all possible rank orders, as we discuss in the next section.

\subsection{Segments and Mixtures of Choice Models}\label{sec:mixtures}
We model customers as each coming from one of $K$ segments $k=1,\ldots,K$, each with its own choice model parameters $\boldsymbol{\phi}^k$ and $\tau^k$. Let $\boldsymbol{\theta}^{\sigma}$ be the customer segment distribution for store $\sigma$, with $\theta^{\sigma}_k$ the probability that an arrival at store $\sigma$ belongs to segment $k$, $\theta^{\sigma}_k \geq 0$, and $\sum_{k=1}^K \theta^{\sigma}_k = 1$. As with other variables, we denote the collection of segment distributions across all stores as $\boldsymbol{\theta}$. Similarly, we denote the collections of choice model parameters across all segments as $\boldsymbol{\phi}$ and $\boldsymbol{\tau}$.

For the nonparametric choice model, each of these segments would have a different rank ordering of items and multiple segments are required in order to have a diverse set of preferences. For the MNL and exogenous choice models, customer segments can be used to borrow strength across multiple stores. All stores share the same underlying segment parameters $\boldsymbol{\phi}$ and $\boldsymbol{\tau}$, but each store's arrivals are represented by a different mixing of these segments, $\boldsymbol{\theta}^{\sigma}$. This model allows us to use data from all of the stores for inferring the choice model parameters, while still allowing stores to differ from each other by having a different mixture of segments.

With the nonparametric choice model, using a segment for each ordered subset of $\{1,\ldots,n\}$ would likely result in more parameters than could be reasonably inferred for $n$ even moderately large. Our inference procedure would be most appropriate for nonparametric models with one or two substitutions (that is, ordered subsets of size 2 or 3), which could still capture a wide range of behaviors.

\subsection{The Likelihood Model}\label{sec:likelihood}
We now describe in detail the generative model for how customer segments, choice models, stock levels, and the arrival function all interact to create transaction data. Consider store $\sigma$ and time period $l$. Customers arrive according to the NHPP for this store. Let $\tilde{t}^{\sigma,l}_1, \ldots, \tilde{t}^{\sigma,l}_{\tilde{m}^{\sigma,l}}$ represent all of the arrival times; these are unobserved, as they may include no-purchases. Each arrival has probability $\theta^{\sigma}_k$ of belonging to segment $k$. They then purchase an item or leave as no-purchase according to the choice model $f_i$. If the $j$'th arrival purchases an item then we observe that purchase at time $\tilde{t}^{\sigma,l}_j$; if they leave as no-purchase we do not observe that arrival at all. The generative model for the observed data $\boldsymbol{t}$ is thus:

\vspace{5pt}
\noindent For store $\sigma=1,\ldots,S$ and time period $l=1,\ldots,L^{\sigma}:$
\begin{itemize}
\setlength\itemsep{0em}
\item Sample arrivals $\tilde{t}^{\sigma,l}_1,\ldots,\tilde{t}^{\sigma,l}_{\tilde{m}^{\sigma,l}} \sim \textrm{NHPP}(\lambda(t \mid \boldsymbol{\eta}^{\sigma}), T)$.
\item For arrival $j=1,\ldots,\tilde{m}^{\sigma,l}$:
\begin{itemize}
\item Sample segment as $k \sim \textrm{Multinomial}(\boldsymbol{\theta}^{\sigma})$.
\item With probability $f_i(s(\tilde{t}^{\sigma,l}_j \mid \boldsymbol{t}^{\sigma,l}, \boldsymbol{N}^{\sigma,l}),\boldsymbol{\phi}^k, \tau^k)$ purchase item $i$, or no purchase with $i=0$.
\item If item $i>0$ purchased, add the time to $\boldsymbol{t}_i^{\sigma,l}$.
\end{itemize}
\end{itemize}

We denote the probability that an arrival at time $t$ purchases item $i$ as $\pi_i(t) = \sum_{k=1}^K \theta^{\sigma}_k f_{i}(s(t \mid \boldsymbol{t}^{\sigma,l}, \boldsymbol{N}^{\sigma,l}),\boldsymbol{\phi}^k,\tau^k)$. An important quantity for the likelihood is the \textit{observed purchase rate}, which is the arrival rate times the purchase probability:
\begin{equation}\label{eq:obsp}
\tilde{\lambda}^{\sigma,l}_i(t) = \lambda(t \mid \boldsymbol{\eta}^{\sigma}) \pi_i(s(t \mid \boldsymbol{t}^{\sigma,l}, \boldsymbol{N}^{\sigma,l}),\boldsymbol{\phi},\boldsymbol{\tau}, \boldsymbol{\theta}^{\sigma}).
\end{equation}
This is the rate at which customers purchase item $i$, incorporating stock availability and customer choice. The corresponding mean function is $\tilde{\Lambda}^{\sigma,l}_i(0,T) = \int_{0}^T \tilde{\lambda}^{\sigma,l}_i(t) dt$.

The following theorem gives the likelihood function corresponding to this generative model.
\begin{theorem}\label{thm:loglklhd}
The log-likelihood function of $\boldsymbol{t}$ is
\begin{align*}
\log p(\boldsymbol{t} \mid & \boldsymbol{\eta}, \boldsymbol{\theta}, \boldsymbol{\phi}, \boldsymbol{\tau}, \boldsymbol{N}, T) \\
&= \sum_{\sigma=1}^S \sum_{l=1}^{L^{\sigma}} \sum_{i=1}^n \left( \sum_{j=1}^{m^{\sigma,l}_i} \log \left(\tilde{\lambda}^{\sigma,l}_i(t^{\sigma,l}_{i,j}) \right)
- \tilde{\Lambda}^{\sigma,l}_i(0,T) \right).
\end{align*}
\end{theorem}
Remarkably, the result is that which would be obtained if we treated the purchases for each item as independent NHPPs with rate $\tilde{\lambda}_i^{\sigma,l}(t)$, the observed purchase rate from (\ref{eq:obsp}). In reality, they are not independent NHPPs inasmuch as they depend on each other via the stock function $s(t \mid \boldsymbol{t}^{\sigma,l}, \boldsymbol{N}^{\sigma,l})$. The key element of the proof is that while the purchase processes depend on each other, they do not depend on the no-purchase arrivals. The proof is given in the Appendix. Also in the Appendix we show how the mean function $\tilde{\Lambda}^{\sigma,l}_i(0,T)$ can be expressed in terms of $\Lambda(0,T \mid \boldsymbol{\eta}^{\sigma})$ and thus computed efficiently, provided the rate function is integrable.

\subsection{Prior Distributions}\label{sec:prior}
Finally, we specify a prior distribution for each of the latent variables: $\boldsymbol{\eta}$, $\boldsymbol{\theta}$, and  $\boldsymbol{\phi}$ and $\boldsymbol{\tau}$ as required by the choice model. The variables $\boldsymbol{\theta}$, $\boldsymbol{\phi}$, and $\boldsymbol{\tau}$ are all probability vectors, so the natural choice is to assign them a Dirichlet or Beta prior:
\begin{align*}
\boldsymbol{\theta} &\sim \textrm{Dirichlet}(\boldsymbol{\alpha})\\
\boldsymbol{\phi}^k &\sim \textrm{Dirichlet}(\boldsymbol{\beta}),\quad k=1,\ldots,K\\
\tau^k &\sim \textrm{Beta}(\boldsymbol{\gamma}),\quad k=1,\ldots,K.
\end{align*}
In our experiments, we used uniform priors by setting the hyperparameters to vectors of $1$. Similarly, a natural choice for the prior distribution of $\boldsymbol{\eta}$ is a uniform distribution for each element: 
\begin{equation*}
\eta^{\sigma}_v \sim \textrm{Uniform}(\boldsymbol{\delta}^v),\quad v=1,\ldots,|\eta^{\sigma}|,\quad \sigma=1,\ldots,S.
\end{equation*}
In our experiments we chose the interval $\boldsymbol{\delta}^v$ large enough to not be restrictive.% For the Hill rate that we use in our data experiments, $|\eta^{\sigma}| = 3$.

%We can then compute the prior probability as
% \begin{align*}
% p(\boldsymbol{\eta},\boldsymbol{\theta}, \boldsymbol{\phi}, \boldsymbol{\tau} \mid \boldsymbol{\alpha},& \boldsymbol{\beta}, \boldsymbol{\gamma}, \boldsymbol{\delta})\\
% &= p(\boldsymbol{\theta} \mid \boldsymbol{\alpha}) \left( \prod_{k=1}^K p(\boldsymbol{\phi}^k \mid \boldsymbol{\beta}) p(\tau^k \mid \boldsymbol{\gamma}) \right) \left( \prod_{\sigma=1}^S \prod _{v=1}^{|\eta^{\sigma}|} p(\eta^{\sigma}_v  \mid \boldsymbol{\delta}^v) \right)\\
% &\propto \left( \prod_{k=1}^K \left( \left( \theta_k \right)^{\alpha_k-1} \left(\tau^k \right)^{\gamma_1-1} \left(1 - \tau^k \right)^{\gamma_2-1}   \prod_{i=1}^n \left( \phi^k_i \right)^{\beta_i-1} \right) \right) \\
% &\quad\times \left( \prod_{\sigma=1}^S \prod _{v=1}^{|\eta^{\sigma}|} \boldsymbol{1}_{\{\eta^{\sigma}_v \in [\delta^v_1,\delta^v_2]\}}  \right).
% \end{align*}
%With this result and the likelihood function from Theorem \ref{thm:loglklhd} we are now equipped to do posterior inference.

\section{MCMC Sampling}
We use MCMC techniques to simulate posterior samples, specifically the stochastic gradient Riemannian Langevin dynamics (SGRLD) algorithm of \cite{patterson13}. SGRLD was developed for posterior inference in topic models, to which our model is conceptually similar. It uses a stochastic gradient that does not require the full likelihood function to be evaluated in every MCMC iteration, which is critical for doing posterior inference on a potentially very large transaction database. %Also, SGRLD is well suited for variables on the probability simplex, as are $\boldsymbol{\theta}$, $\boldsymbol{\phi}^k$, and $\tau^k$.

We first transform each of the probability variables using the expanded-mean parameterization \cite{patterson13}. Consider the latent variable $\boldsymbol{\theta}$, which has as constraints $\theta_k \geq 0$ and $\sum_{k=1}^K \theta_k = 1$. Take $\boldsymbol{\tilde{\theta}}$ a random variable with support on $\mathbb{R}_{+}^K$, and give $\boldsymbol{\tilde{\theta}}$ a prior distribution consisting of a product of $\textrm{Gamma}(\alpha_k,1)$ distributions:
\begin{equation*}
p(\boldsymbol{\tilde{\theta}} \mid \boldsymbol{\alpha}) \propto \prod_{k=1}^K \tilde{\theta}_k^{\alpha_k-1} \exp(-\tilde{\theta}_k).
\end{equation*}
The posterior sampling is done over variables $\boldsymbol{\tilde{\theta}}$ by mirroring any negative proposal values about $0$. We then set $\theta_k = \tilde{\theta}_k/\sum_{r=1}^K  \tilde{\theta}_r$.
This parameterization is equivalent to sampling on $\boldsymbol{\theta}$ with a $\textrm{Dirichlet}(\alpha)$ prior, but does not require the probability simplex constraint. The same transformation is done to $\boldsymbol{\phi}^k$ and $\tau^k$.

Let $\boldsymbol{z} = \{\boldsymbol{\eta}, \boldsymbol{\tilde{\theta}},\boldsymbol{\tilde{\phi}},\boldsymbol{\tilde{\tau}}\}$ represent the complete collection of transformed latent variables whose posterior distribution we are inferring. From state $\boldsymbol{z}_w$ on MCMC iteration $w$, the next iteration moves to the state $\boldsymbol{z}_{w+1}$ according to
\begin{align*}
\boldsymbol{z}_{w+1} &= \boldsymbol{z}_w \\
&\quad+ \frac{\epsilon_w}{2} \left( \textrm{diag}(\boldsymbol{z}_w) \nabla \log p(\boldsymbol{z}_w \mid \boldsymbol{t}, \boldsymbol{\alpha}, \boldsymbol{\beta}, \boldsymbol{\gamma}, \boldsymbol{\delta}, \boldsymbol{N}, T)  +\boldsymbol{1} \right) \\
&\quad+ \textrm{diag}(\boldsymbol{z}_w)^{\frac{1}{2}} \psi,\\
\psi &\sim \mathcal{N}(0,\epsilon_w I).
\end{align*}
The iteration performs a gradient step plus normally distributed noise, using the natural gradient of the log posterior, which is the manifold direction of steepest descent using the metric $G(\boldsymbol{z}) = \textrm{diag}(\boldsymbol{z})^{-1}$. Using Bayes' theorem, the posterior gradient can be decomposed into the likelihood gradient and the prior gradient, and
% \begin{equation*}
% \nabla \log p(\boldsymbol{z}_w \mid \boldsymbol{t}, \boldsymbol{\alpha}, \boldsymbol{\beta}, \boldsymbol{\gamma}, \boldsymbol{\delta}, \boldsymbol{N}, T) = \nabla \log p(\boldsymbol{t} \mid \boldsymbol{z}_w,\boldsymbol{N}, T) + \nabla \log p(\boldsymbol{z}_w \mid \boldsymbol{\alpha}, \boldsymbol{\beta}, \boldsymbol{\gamma}, \boldsymbol{\delta}).
% \end{equation*}
we use a stochastic gradient approximation for the likelihood gradient. On MCMC iteration $w$, rather than use all $L^{\sigma}$ time periods to compute the gradient we use a uniformly sampled collection of time periods $\mathcal{L}^{\sigma}_w$. The gradient approximation is then
\begin{align*}
\nabla & \log p(\boldsymbol{t} \mid \boldsymbol{z}_w,\boldsymbol{N}, T) \\
&\approx \sum_{\sigma=1}^S \frac{L^{\sigma}}{|\mathcal{L}_w^{\sigma}|} \sum_{l \in \mathcal{L}_w^{\sigma}} \sum_{i=1}^n \nabla \left( \sum_{j=1}^{m^{\sigma,l}_i} \log \left(\tilde{\lambda}^{\sigma,l}_i(t^{\sigma,l}_{i,j}) \right)
- \tilde{\Lambda}^{\sigma,l}_i(0,T) \right).
\end{align*}
The iterations will converge to the posterior samples if the step size schedule is chosen such that $\sum_{w=1}^{\infty} \epsilon_w = \infty$ and $\sum_{w=1}^{\infty} \epsilon^2_w < \infty$ \cite{welling11}. In our simulations and experiments we used three time periods for the stochastic gradient approximations. We followed \cite{patterson13} and took $\epsilon_w = a((1+q/b)^{-c})$, with parameters $a$, $b$, and $c$ chosen using cross-validation over a grid to minimize out-of-sample perplexity. We drew 10,000 samples from each of three chains initialized at a local maximum  \textit{a posteriori} solution found from a random sample from the prior. We verified convergence using the Gelman-Rubin diagnostic after discarding the first half of the samples as burn-in \cite{Gelman92}, and then merged samples from all three chains to estimate the posterior.

\begin{figure}[t]
\centering
\includegraphics[scale=0.88]{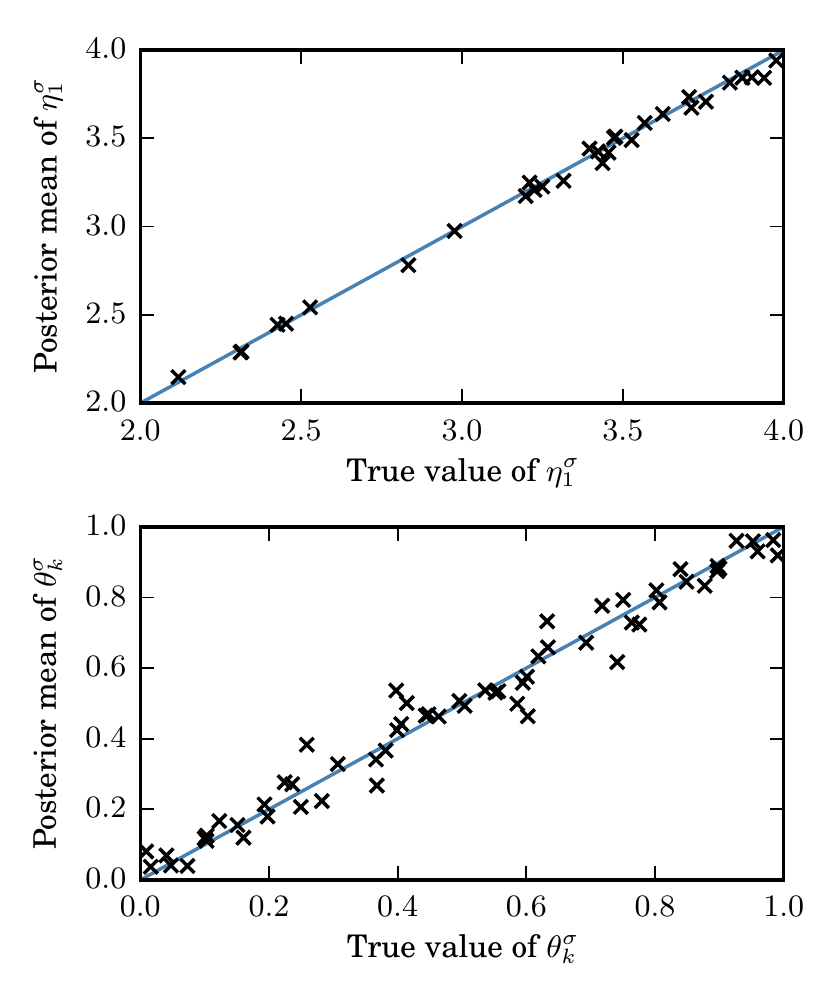}
\caption{Markers in the top panel show, for each randomly chosen value of $\eta^{\sigma}_1$ used in the set of simulations (3 stores $\times$ 10 simulations), the corresponding estimate of the posterior mean. The bottom panel shows the same result for each value of $\theta^{\sigma}_k$ used (3 stores $\times$ 2 segments $\times$ 10 simulations). For a range of generating parameter values, the posterior distributions were centered on the true values.} \label{fig:sim1_5}
\end{figure}

\begin{figure}[t]
\centering
\includegraphics[scale=0.88]{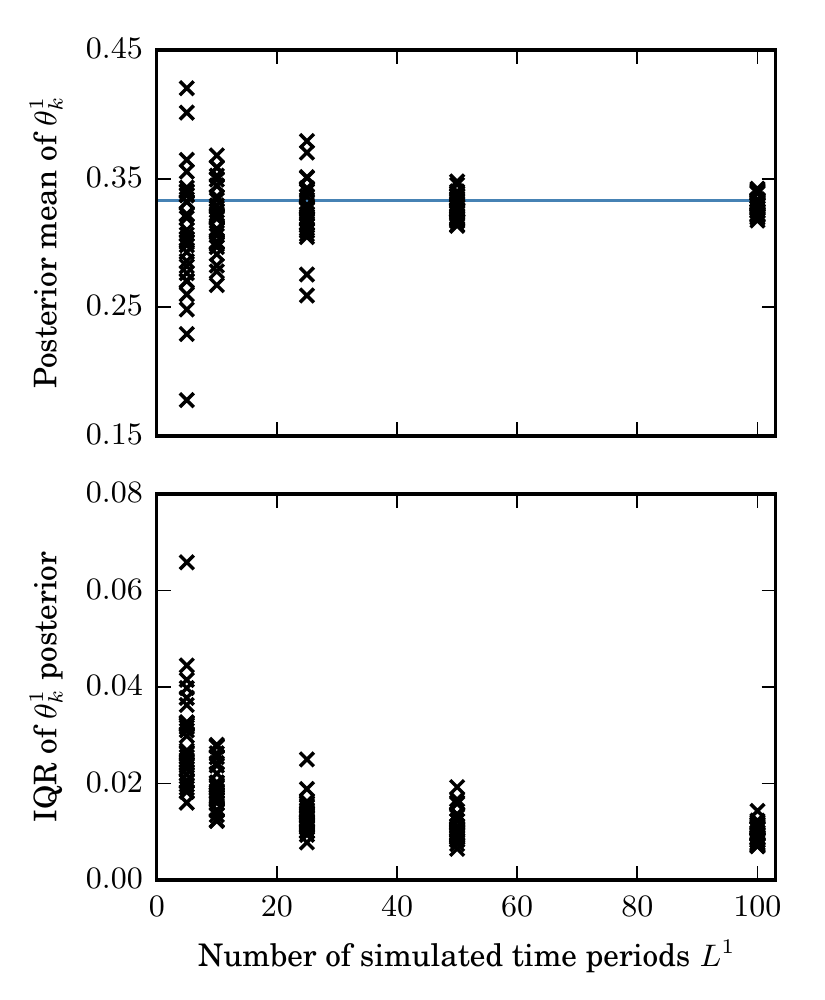}
\caption{Each marker corresponds to the posterior distribution for $\theta^1_k$ from a simulation with the corresponding number of time periods, across the 3 values of $k$ where the true value equaled $0.33$. The top panel shows the posterior mean for each of the simulations across the different number of time periods. The bottom panel shows the interquartile range (IQR) of the posterior. As the amount of available data increased, the posterior distributions became increasingly concentrated on the true values.} \label{fig:sim3_3}
\end{figure}

\section{Simulation Study}
We use a collection of simulations to illustrate and analyze the model and the inference procedure. We use a variety of rate functions and choice models throughout the simulations to demonstrate this flexibility of the model. First we use the simulations to verify that the posterior concentrates around the true generating values for a wide selection of arrival rate functions, choice models, and model parameters. Then we use simulations to investigate the dependence on the amount of data used in the inference. The simulations show that the posterior variance decreases as the size of the training data set increases, which is remarkable inasmuch as the reduction of uncertainty came with no additional computational cost because of the stochastic gradient approximation for the likelihood.

The first set of simulations used the homogeneous rate function $\lambda(t \mid \boldsymbol{\eta}^{\sigma}) = \eta^{\sigma}_1$ and the exogenous choice model given in (\ref{eq:exo}), with $S=3$ stores, $K=2$ segments, and $n=3$ items. The choice model parameters were fixed at $\tau^1 = \tau^2 = 0.75$, $\boldsymbol{\phi}^1 = [0.75, 0.2, 0.05]$, and  $\boldsymbol{\phi}^2 = [0.33, 0.33, 0.34]$. For each of 10 simulated data sets, the segment distributions $\boldsymbol{\theta}^\sigma$ were chosen independently at random from a uniform Dirichlet distribution and the arrival rates $\eta^{\sigma}_1$ were chosen independently at random from a uniform distribution on $[2,4]$. For each store, we simulated $25$ time periods, each of length $T=1000$ and with the initial stock for each item chosen uniformly between $0$ and $500$, independently at random for each item, time period, and store. Purchase data were then generated according to the generative model in Section \ref{sec:likelihood}. Figure \ref{fig:sim1_5} shows the posterior means estimated from the MCMC samples across the $10$ repeats of the simulation, each with different segment distributions and rate parameters. This figure shows that across the full range of parameter values used in these simulations the posterior mean was close to the true, generating value.

In the second set of simulations we used the Hill rate function with the nonparametric choice model, with 3 items. We used all sets of preference rankings of size $1$ and $2$, which for $3$ items requires a total of $9$ segments. We simulated data for a single store, with the segment proportion $\theta^1_k$ set to $0.33$ for preference rankings $\{1\}$, $\{1, 2\}$, and $\{3, 2\}$: The first segment prefers item $1$ and will leave with no purchase if item $1$ is not available, the second segment prefers item $1$ but is willing to substitute to item $2$, and the third segment prefers item $3$ but is willing to substitute to item $2$. The segment proportions for the remaining $6$ preference rankings were set to zero. With this simulation we study the effect of the number of time periods used in the inference, $L^1$. $L^1$ was taken from $\{5,10,25,50,100\}$, and for each of these values 10 simulations were done.

As in Figure \ref{fig:sim1_5}, the posterior densities for the segment proportions were concentrated near their true values. Figure \ref{fig:sim3_3} shows how the posteriors depended on the number of time periods of available data. The top panel shows that the posterior means for the non-zero segment proportions tended closer to the true value as more data were made available. The bottom panel shows the actual concentration of the posterior, where the interquartile range of the posterior decreased with the number of time periods. Because we use a stochastic gradient approximation, using more time periods came at no additional computational cost: We used 3 time periods for each gradient approximation regardless of the available number.

\section{Case Study: Bakery Sales}\label{sec:dataexp}
We now provide the results of the model applied to real transaction data. As part of our case study, we evaluate the predictive power of the model and sample the posterior distribution of lost sales due to stockouts.

We obtained one semester of sales data from the bakery at 100 Main Marketplace, a cafe located at MIT, for a collection of cookies: oatmeal, double chocolate, and chocolate chip. The data set included all purchase times for 151 days; we treated each day as a time period (11:00 a.m. to 7:00 p.m.), and there were a total of 4084 purchases. Stock data were not available, only purchase times, so for the purpose of these experiments we set the initial stock for each time period equal to the number of purchases for the time period - thus every item was treated as stocked out after its last recorded purchase. This may be a reasonable assumption for these cookies given that they are perishable baked goods which are meant to stock out by the end of the day, but in any case the experiments still provide a useful illustration of the method.

The empirical purchase rate for the cookies, shown in Figure \ref{fig:lnch_2}, was markedly nonhomogeneous: there is a broad peak at lunch time and two sharp peaks at common class ending times. We modeled the rate function with a combination of the Hill function $\lambda^{\textrm{H}}(t)$ (\ref{eq:hill}) and a fixed function consisting of only two peaks at the two afternoon peak times, $\lambda^{\textrm{p}}(t)$, obtained via a spline. The Hill function has three parameters, and then a fourth parameter provided the weight of the fixed peaks that were added in: $\lambda(t \mid \boldsymbol{\eta}) = \lambda^{\textrm{H}}(t \mid \eta_1, \eta_2, \eta_3) + \eta_4 \lambda^{\textrm{p}}(t)$. We fit the model separately with the exogenous and nonparametric choice models.

Figure \ref{fig:lnch_2} shows 20 posterior samples for the model's predicted average purchase rate over all time periods, which equals $\frac{1}{151} \sum_{l=1}^{151} \sum_{i=1}^3 \tilde{\lambda}^{1,l}_i$, from the fit with the nonparametric choice model. These samples show that the model provides an accurate description of the arrival rate. The variance in the samples provides an indication of the uncertainty in the model, which further motivates the use of the posterior predictive distribution over a point estimate for making predictions.

\begin{figure}[t]
\centering
\includegraphics[scale=0.88]{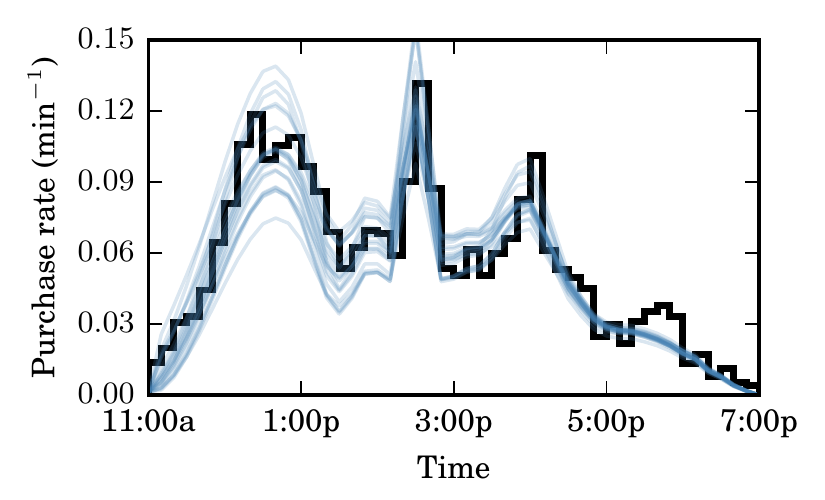}
\caption{A normalized histogram of purchase times for the cookies, across time periods, along with posterior samples for the model's corresponding predicted purchase rate.} \label{fig:lnch_2}
\end{figure}

Figure \ref{fig:lnch_3} shows the posterior density for the substitution rate $\boldsymbol{\tau}$, obtained by fitting the model with the exogenous choice model. The substitution rate is very low, indicating that most customers left without a purchase if their preferred cookie was not in stock. The posterior distribution of the item preference vector is given in \ref{fig:lnch_4}. Chocolate chip cookies were the strong favorite, followed by double chocolate and lastly oatmeal. 

\begin{figure}[t]
\centering
\includegraphics[scale=0.88]{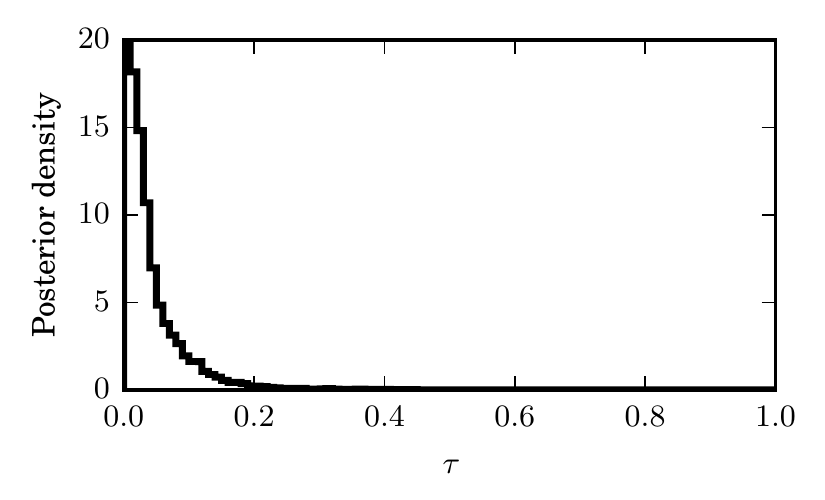}
\caption{Normalized histogram of posterior samples of the exogenous choice model substitution rate, for the cookie data. This is the probability that a customer will substitute if his or her preferred item is out of stock.} \label{fig:lnch_3}
\end{figure}

\begin{figure*}[h]
\centering
\includegraphics[scale=0.88]{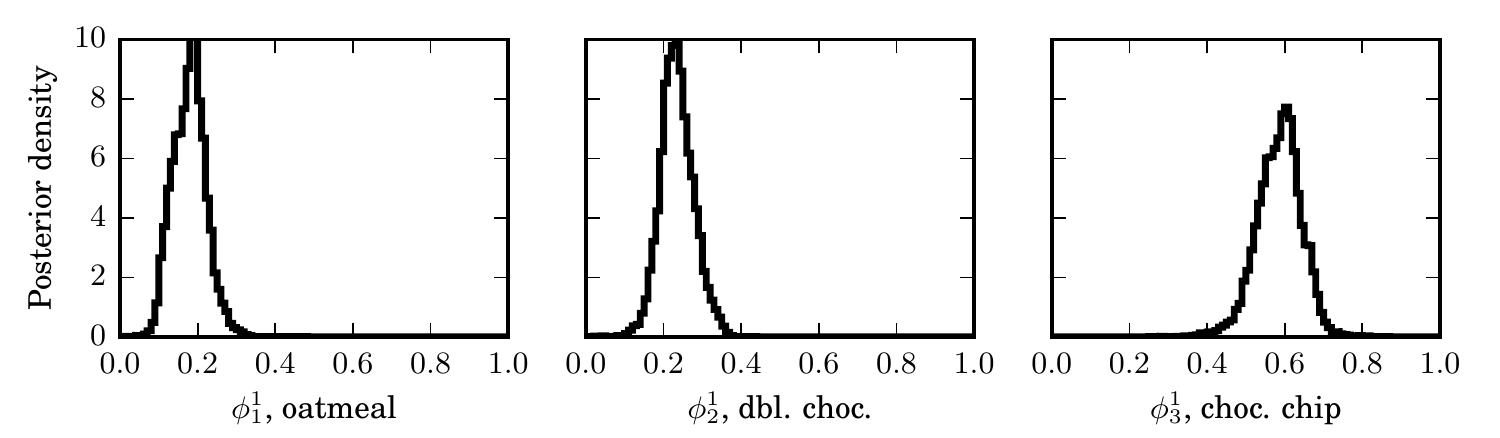}
\caption{Normalized histograms of posterior samples of the exogenous choice model preference vector, for the cookie data. This vector provides the probability that each item is a customer's primary choice.} \label{fig:lnch_4}
\end{figure*}

\begin{figure*}[t]
\centering
\includegraphics[scale=0.88]{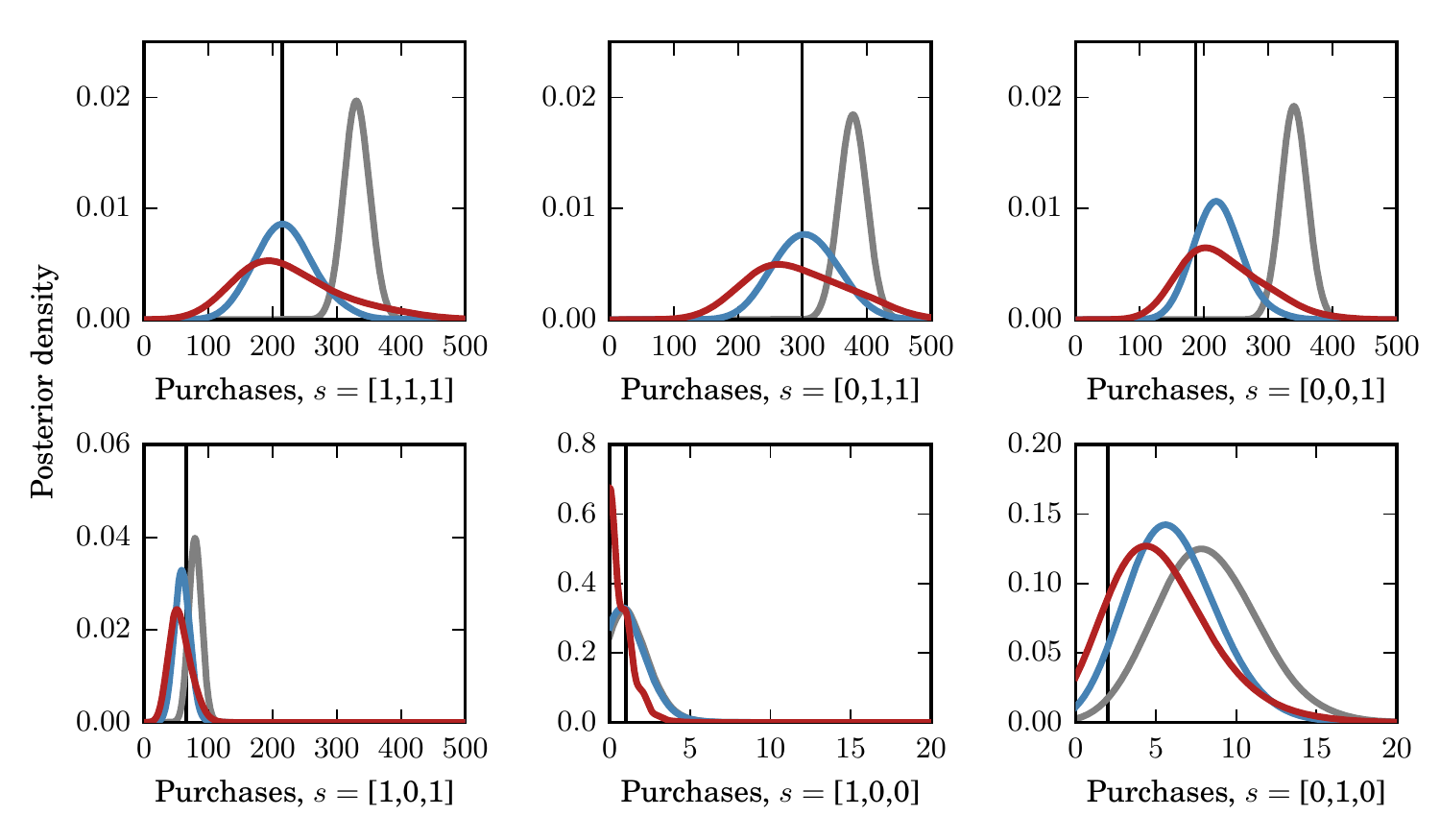}
\caption{Smoothed posterior densities for the number of purchases during test set intervals with the indicated stock availability for cookies [oatmeal, double chocolate, chocolate chip]. The density in blue is for the nonparametric choice, red is for the exogenous choice, and gray is for the baseline homogeneous arrival rate with MNL choice. The vertical line indicates the true value.} \label{fig:pred_1}
\end{figure*}

\subsection{Predictive Performance}\label{sec:predictions}
The next set of experiments establish that the model has predictive power on real data. We evaluated the predictive power of the model by predicting out-of-sample purchase counts during periods of varying stock availability. We took the first $80\%$ of time periods (120 time periods) as training data and did posterior inference. The latter 31 time periods were held out as test data, the goal being to use data from the first part of the semester to make predictions about the latter part. We considered each possible level of stock unavailability, \textit{i.e.}, $s=[1, 0, 0]$, $s=[0, 1, 0]$, etc. For each stock level, we found all of the time intervals in the test periods with that stock. The prediction task was, given only the time intervals and the corresponding stock level, to predict the total number of purchases that took place during those time intervals in the test periods. The actual number of purchases is known and thus predictive performance can be evaluated. There were no intervals where only chocolate chip cookies were out of stock, but predictions were made for every other stock combination.

This is a meaningful prediction task because good performance requires being able to accurately model both the arrival rate as a function of time and how the actual purchases then depend on the stock. We compare predictive performance to a baseline model that has previously been proposed for this problem by \cite{Vulcano12}, which is the maximum likelihood model with a homogeneous arrival rate and the MNL choice model. We discuss this and other related works in more detail in Section \ref{sec:relatedwork}.

For the MNL baseline, the parameter $\tau^1$ is unidentifiable and cannot be estimated. We fit the model for each fixed $\tau^1 \in \{0.1, 0.2, \ldots, 0.9\}$, and show here the results with the value of $\tau^1$ that minimized the out-of-sample absolute deviation between the model expected number of purchases and the true number of purchases, which was $0.4$. That is, we show here the results that would have been obtained if we had known \textit{a priori} the best value of $\tau^1$, and thus show the best possible performance of the baseline.

For our model, for each choice model (nonparametric and exogenous) posterior samples obtained from the MCMC procedure were used to estimate the posterior predictive distribution for the number of purchases under each stock level. For the maximum likelihood baseline, we used simulation to estimate the distribution of purchase counts conditioned on the point estimate model. These posterior densities, smoothed with a kernel density estimate, are given in Figure \ref{fig:pred_1}. Despite their very different natures, the predictions made by the exogenous and  nonparametric models are quite similar, and both have posterior means close to the true values for all stock levels. The baseline maximum likelihood model with a homogeneous arrival rate and MNL choice performs very poorly.

\subsection{Lost Sales Due to Stockouts}

\begin{figure*}[t]
\centering
\includegraphics[scale=0.88]{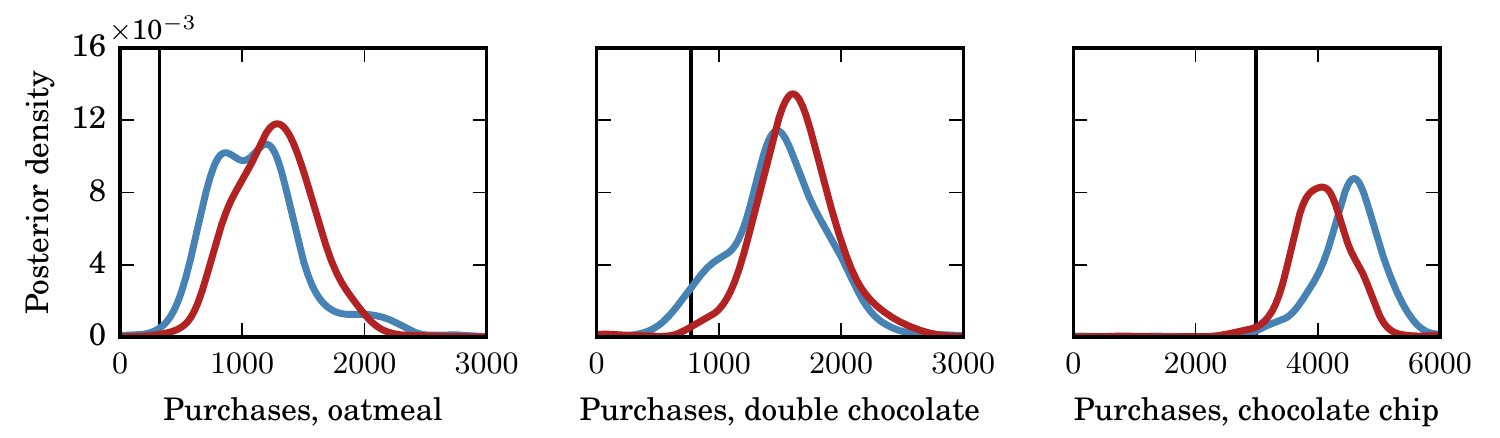}
\caption{Smoothed posterior densities for the number of purchases during all periods, had there had been no stockouts. The blue density is the result with the nonparametric choice model, and the red with the exogenous. The vertical line indicates the number of purchases in the data.} \label{fig:pred_4}
\end{figure*}

Our purpose in inferring the model is to use it to make better stocking decisions. An important starting point is to use the inferred parameters to estimate what the sales would have been had there not been any stockouts. This allows us to know how much revenue is being lost with our current stocking strategy. We estimated posterior densities for the number of purchases of each item across 151 time periods, with full stock. Figure \ref{fig:pred_4} compares those densities to the actual number of cookie purchases in the data.

For each of the cookies, the actual number of purchases was significantly less than the posterior density for purchases with full stock, indicating that there were substantial lost sales due to stockouts. With the nonparametric model, the difference between the full-stock posterior mean and the actual number of purchases was 791 oatmeal cookies, 707 double chocolate cookies, and 1535 chocolate chip cookies. Figure \ref{fig:lnch_3} shows that customers were generally unwilling to substitute, which would have contributed to the lost sales.

\section{Related Work}\label{sec:relatedwork}
The primary work on this problem, estimating demand and substitution from sales transaction data with stockouts and unobserved no-purchases, was done by \cite{Vulcano12}. They model customer arrivals using a homogeneous Poisson process within each time period, meaning the arrival rate is constant throughout each time period. Customers then choose an item, or an unobserved no-purchase, according to the MNL choice model. They derive an EM algorithm to solve the corresponding maximum likelihood problem. In the prediction task of Section \ref{sec:predictions} we compared our results with this model as the baseline and found that it was unable to make accurate predictions with our case study data. Our model overcomes several limitations of this model, thereby substantially advancing the power of the inference and the settings in which the model can be used. First, Figure \ref{fig:lnch_2} shows that the arrivals are significantly nonhomogeneous throughout the day, and modeling the arrival rate as constant throughout the day is likely the reason the baseline model failed the prediction task. The work in \cite{Vulcano12} proposes extending their model to a nonhomogeneous setting by choosing sufficiently small time periods that the arrival rate can be approximated as piecewise constant. However, with the level of nonhomogeneity seen in Figure \ref{fig:lnch_2} it is implausible that accurate estimation could be done for the number of segments (and thus separate rate parameters) required to model the arrival rate with a piecewise-constant function. Second, our model does not require using the MNL choice model, which avoids the issue with the parameter $\tau$ being unidentifiable. This parameter represents the proportion of arrivals that do not purchase anything even when all items are in stock, and is not something that a retailer would necessarily know. Finally, we take a Bayesian approach to inference and produce posterior predictive distributions. This becomes especially important in this setting where the parameters themselves are of secondary interest to using the model to make predictions about lost revenue and to make decisions about stocking strategies. 

Other work in this area includes \cite{Anupindi98}, where customer arrivals are modeled with a homogeneous Poisson process and purchase probabilities are modeled explicitly for each stock combination, as opposed to using a choice model. Their model does not scale well to a large number of items as the likelihood expression includes all stock combinations found in the data. The work of \cite{Vulcano12} is extended in \cite{Vulcano14} to incorporate nonparametric choice models, for which maximum likelihood estimation becomes a large-scale concave program that must be solved via a mixed integer program subproblem. There is a large body of work on estimating demand and choice in settings different than that which we consider here, such as discrete time \cite{Talluri01, Vulcano10}, panel or aggregate sales data \cite{Campoa03, Kalyanam07, Musalem10}, negligible no purchases \cite{Kok07}, and online learning with simultaneous ordering decisions \cite{Jain15}. These models and estimation procedures do not apply to the setting that we consider here, which is retail transaction data with stockouts and unobserved no-purchases; \cite{Jain15} provide a review of the various threads of research in the larger field of demand and choice estimation.

Our work fits into a growing body of work in advancing the use of statistics in areas of business. These areas include marketing \cite{Green66, Soriano13, Banks06}, market analysis \cite{Finazzi13, Rubin06}, demand forecasting \cite{Liu01, Shen08}, and pricing \cite{Ghose06, letham14}. These works, and ours, address a real need for rigorous statistical methodologies in business, as well as a substantial opportunity for impact.

\section{Discussion}
We have developed a Bayesian model for inferring primary demand and consumer choice in the presence of stockouts. The model can incorporate a realistic model of the customer arrival rate, and is flexible enough to handle a variety of different choice models. Our model is conceptually related to topic models like latent Dirichlet allocation \cite{Blei03}. Variants of topic models are regularly applied to very large text corpora, with a large body of research on how to effectively infer these models. That research was the source of the stochastic gradient MCMC algorithm that we used, which allows inference from even very large transaction databases.

The simulation study showed that when data were actually generated from the model, we were able to recover the true generating values. It further showed that the posterior bias and variance decreased as more data were made available, an improvement that came without any additional computational cost due to the stochastic gradient.

In the case study we applied the model and inference to real sales transaction data from a local bakery. The daily purchase rate in the data was clearly nonhomogeneous, with several peak periods. These data clearly demonstrated the importance of modeling nonhomogeneous arrival rates in retail settings. In a prediction task that required accurate modeling of both the arrival rate and the choice model, we showed that the model was able to make accurate predictions and significantly outperform the baseline approach.

Finally, we showed how the model can be used to estimate lost sales due to stockouts. The posterior provided evidence of substantial lost cookie sales. The model and inference procedure we have developed provide a new level of power and flexibility that will aid decision makers in using transaction data to make smarter decisions.

\section{Acknowledgements}
We are grateful to the staff at 100 Main Marketplace at the Massachusetts Institute of Technology who provided data for this study.

%\bibliographystyle{abbrv}
%\bibliography{demand} 

\appendix
\allowdisplaybreaks
Here we provide the derivation of the log-likelihood function. The proof will use two results which themselves are straightforward to show.

\begin{proposition}
\begin{equation*}
p(t_j \mid  t_{j-1}, \boldsymbol{\eta}^{\sigma}) = \exp(-\Lambda(t_{j-1},t_j \mid \boldsymbol{\eta}^{\sigma})) \lambda(t_j \mid \boldsymbol{\eta}^{\sigma}).
\end{equation*}
\end{proposition}

\begin{proposition}
\begin{equation*}
\Lambda(0,T \mid \boldsymbol{\eta}^{\sigma})  = \sum_{i=0}^n \tilde{\Lambda}_i^{\sigma,l}(0,T).
\end{equation*}
\end{proposition}

\begin{proof}[of Theorem 1]
We consider the complete arrivals $\tilde{\boldsymbol{t}}^{\sigma,l}$, which include both the observed arrivals $\boldsymbol{t}^{\sigma,l}$ as well as the unobserved arrivals that left as no-purchase, which we here denote $\boldsymbol{t}_0^{\sigma,l} = \left\{\boldsymbol{t}_{0,j}^{\sigma,l}\right\}_{j=1}^{m^{\sigma,l}_0}$. We define an indicator $\tilde{I}^{\sigma,l}_j$ equal to $i$ if the customer at time $\tilde{t}^{\sigma,l}_j$ purchased item $i$, or $0$ if this customer left as no-purchase. For store $\sigma$ and time period $l$, 
\begin{align*}
p(\boldsymbol{t}^{\sigma,l}_0,&\boldsymbol{t}^{\sigma,l} \mid \boldsymbol{\eta}^{\sigma}, \boldsymbol{\theta}^{\sigma}, \boldsymbol{\phi}, \boldsymbol{\tau}, \boldsymbol{N}, T)\\
&=  \mathbb{P}\left(\textrm{no arrivals in }\left(\tilde{t}^{\sigma,l}_{\tilde{m}^{\sigma,l}},T\right] \mid \boldsymbol{\tilde{t}^{\sigma,l}}, \boldsymbol{\eta}^{\sigma}\right) \\
&\quad \times
\prod_{j=1}^{\tilde{m}^{\sigma,l}} p(\tilde{t}^{\sigma,l}_j \mid \tilde{t}^{\sigma,l}_{<j}, \boldsymbol{\eta}^{\sigma}) p(\tilde{I}^{\sigma,l}_j \mid \tilde{t}^{\sigma,l}_{<j}, \boldsymbol{\theta}^{\sigma}, \boldsymbol{\phi}, \boldsymbol{\tau}, \boldsymbol{N})\\
% &=  \exp(-\Lambda(\tilde{t}_{\tilde{m}^{\sigma,l}},T \mid \boldsymbol{\eta}^{\sigma})) \prod_{j=1}^{\tilde{m}^{\sigma,l}} p(\tilde{t}^{\sigma,l}_j \mid \tilde{t}^{\sigma,l}_{j-1}, \boldsymbol{\eta}^{\sigma}) \pi_{\tilde{I}^{\sigma,l}_j}(\tilde{t}^{\sigma,l}_j)\\
&= \exp(-\Lambda(\tilde{t}_{\tilde{m}^{\sigma,l}},T \mid \boldsymbol{\eta}^{\sigma})) \\
&\quad \times
\lambda(\tilde{t}^{\sigma,l}_1 \mid \boldsymbol{\eta}^{\sigma})\exp(-\Lambda(0,\tilde{t}^{\sigma,l}_1 \mid \boldsymbol{\eta}^{\sigma})) \pi_{\tilde{I}^{\sigma,l}_1}(\tilde{t}^{\sigma,l}_1)\\
&\quad \times
\prod_{j=2}^{\tilde{m}^{\sigma,l}} \lambda(\tilde{t}^{\sigma,l}_j \mid \boldsymbol{\eta}^{\sigma})\exp(-\Lambda(\tilde{t}^{\sigma,l}_{j-1},\tilde{t}^{\sigma,l}_j \mid \boldsymbol{\eta}^{\sigma})) \pi_{\tilde{I}^{\sigma,l}_j}(\tilde{t}^{\sigma,l}_j)\\
&= \exp(-\Lambda(0,T \mid \boldsymbol{\eta}^{\sigma})) \prod_{i=0}^n \prod_{j: \tilde{I}^{\sigma,l}_j = i} \lambda(\tilde{t}^{\sigma,l}_j \mid \boldsymbol{\eta}^{\sigma}) \pi_{i}(\tilde{t}^{\sigma,l}_j )\\
&= \exp(-\Lambda(0,T \mid \boldsymbol{\eta}^{\sigma})) \prod_{i=0}^n \prod_{j=1}^{m^{\sigma,l}_i} \lambda(t^{\sigma,l}_{i,j} \mid \boldsymbol{\eta}^{\sigma}) \pi_{i}(t^{\sigma,l}_{i,j})\\
&= \left(\exp(-\tilde{\Lambda}_0^{\sigma,l}(0,T)) \prod_{j=1}^{m^{\sigma,l}_0} \tilde{\lambda}_0^{\sigma,l}(t^{\sigma,l}_{0,j})  \right) \\
&\quad \times
\left(  \prod_{i=1}^n \exp(-\tilde{\Lambda}_i^{\sigma,l}(0,T)) \prod_{j=1}^{m^{\sigma,l}_i} \tilde{\lambda}_i^{\sigma,l}(t^{\sigma,l}_{i,j}) \right).
\end{align*}
We have then that
\begin{align*}
p( \boldsymbol{t}^{\sigma,l} \mid \boldsymbol{\eta}^{\sigma},\boldsymbol{\theta}^{\sigma},& \boldsymbol{\phi}, \boldsymbol{\tau}, \boldsymbol{N}, T) \\
&= \int p(\boldsymbol{t}^{\sigma,l}_0, \boldsymbol{t}^{\sigma,l} \mid \boldsymbol{\eta}^{\sigma}, \boldsymbol{\theta}^{\sigma}, \boldsymbol{\phi}, \boldsymbol{\tau}, \boldsymbol{N}, T) d\boldsymbol{t}_0^{\sigma,l}\\
&=  \left( \int \exp(-\tilde{\Lambda}_0^{\sigma,l}(0,T)) \prod_{j=1}^{m^{\sigma,l}_0} \tilde{\lambda}_0^{\sigma,l}(t^{\sigma,l}_{0,j})  d\boldsymbol{t}_0^{\sigma,l} \right) \\
&\quad\times
\left(  \prod_{i=1}^n \exp(-\tilde{\Lambda}_i^{\sigma,l}(0,T)) \prod_{j=1}^{m^{\sigma,l}_i} \tilde{\lambda}_i^{\sigma,l}(t^{\sigma,l}_{i,j})  \right)\\
&=  \prod_{i=1}^n \exp(-\tilde{\Lambda}_i^{\sigma,l}(0,T)) \prod_{j=1}^{m^{\sigma,l}_i} \tilde{\lambda}_i^{\sigma,l}(t^{\sigma,l}_{i,j}),
\end{align*}
since the last integrand is exactly the joint density for the arrivals from an NHPP with rate $\tilde{\lambda}_0^{\sigma,l}(t)$, and so integrates to $1$.

Given the model parameters, data are generated independently for each $\sigma$ and $l$, thus
\begin{align*}
\log p(\boldsymbol{t} \mid & \boldsymbol{\eta},\boldsymbol{\theta}, \boldsymbol{\phi}, \boldsymbol{\tau}, \boldsymbol{N}, T) \\
&= \sum_{\sigma=1}^S \sum_{l=1}^{L^{\sigma}} \log p(\boldsymbol{t}^{\sigma,l} \mid \boldsymbol{\eta}^{\sigma},\boldsymbol{\theta}^{\sigma}, \boldsymbol{\phi}, \boldsymbol{\tau}, \boldsymbol{N}, T)\\
&= \sum_{\sigma=1}^S \sum_{l=1}^{L^{\sigma}} \sum_{i=1}^n \left( \sum_{j=1}^{m^{\sigma,l}_i} \log \left(\tilde{\lambda}^{\sigma,l}_i(t^{\sigma,l}_{i,j}) \right)
- \tilde{\Lambda}^{\sigma,l}_i(0,T) \right).
\end{align*}
\end{proof}

We now show how $\tilde{\Lambda}^{\sigma,l}_i(0,T)$ can be expressed analytically in terms of $\Lambda(0,T \mid \boldsymbol{\eta}^{\sigma})$. For convenience, in this section we suppress in the notation the dependence of the stock on past arrivals and initial stock levels and will write $s(t \mid \boldsymbol{t}^{\sigma,l}, \boldsymbol{N}^{\sigma,l})$ as simply $s(t)$. We consider each of the time intervals where the stock $s(t)$ is constant. Let the sequence of times $q^{\sigma,l}_1,\ldots,q^{\sigma,l}_{Q^{\sigma,l}}$ demarcate the intervals of constant stock. That is, $[0,T] = \bigcup_{r=1}^{Q^{\sigma,l}-1} [q^{\sigma,l}_r,q^{\sigma,l}_{r+1}]$ and $s(t)$ is constant for $t\in[q^{\sigma,l}_r,q^{\sigma,l}_{r+1})$ for $r=1,\ldots,Q^{\sigma,l}-1$. Then,
\begin{align*}
\tilde{\Lambda}^{\sigma,l}_i&(0,T) \\
&= \int_{0}^T \tilde{\lambda}^{\sigma,l}_i(t) dt\\
&= \int_{0}^T  \lambda(t \mid \boldsymbol{\eta}^{\sigma}) \sum_{k=1}^K \theta^{\sigma}_k f_i(s(t),\boldsymbol{\phi}^k,\tau^k) dt\\
&= \sum_{r=1}^{Q^{\sigma,l}-1} \left( \int_{q^{\sigma,l}_r}^{q^{\sigma,l}_{r+1}} \lambda(t \mid \boldsymbol{\eta}^{\sigma}) \sum_{k=1}^K \theta^{\sigma}_k f_i(s(q^{\sigma,l}_r),\boldsymbol{\phi}^k,\tau^k) dt \right)\\
&= \sum_{r=1}^{Q^{\sigma,l}-1} \left(\sum_{k=1}^K \theta^{\sigma}_k f_i(s(q^{\sigma,l}_r ),\boldsymbol{\phi}^k,\tau^k) \right)\Lambda(q^{\sigma,l}_r,q^{\sigma,l}_{r+1} \mid \boldsymbol{\eta}^{\sigma}).
\end{align*}
With this formula, the likelihood function can be computed for any parameterization $\lambda(t \mid \boldsymbol{\eta}^{\sigma})$ desired so long as it is integrable.

\end{document}